\documentclass{article}
\usepackage{fullpage}

\usepackage[cmex10]{amsmath}
\usepackage{algorithmic}
\usepackage{array}
\usepackage{mdwmath}
\usepackage{mdwtab}
\usepackage[tight,footnotesize]{subfigure}
\usepackage[font=footnotesize]{subfig}
\usepackage{color}
\usepackage{algorithm}
\usepackage{afterpage}  
\usepackage{float}      
\usepackage{graphics}
\usepackage{lscape}     
\usepackage{psfrag}     
\usepackage[cmex10]{amsmath}    
\usepackage{longtable}  
\usepackage{fancyhdr}   
\usepackage{enumerate}  
\usepackage{amssymb}    
\usepackage{setspace}   
\usepackage{makeidx}    
\usepackage{multirow}
\usepackage[tight,footnotesize]{subfigure}
\usepackage{cite}
\usepackage{amsmath,amssymb,amsthm,color}
\usepackage{amsmath,amssymb,color}
\usepackage{cite,graphicx,psfrag,url,setspace,lineno}
\usepackage{enumerate}

\newcommand{\w}{\mathrm{w}}
\renewcommand{\v}{\mathrm{v}}

\newcommand{\R}{\mathbb{R}}

 \newcommand{\sgn}{\operatorname{sign}}

\newcommand{\Sset}{(T\cap\widetilde{T}^c)\cup(T^c\cap\widetilde{T})}
\newtheorem{thm}{Theorem}[]
\newtheorem{lemma}[thm]{Lemma}
\newtheorem{prop}[thm]{Proposition}
\newtheorem{cor}[thm]{Corollary}
\theoremstyle{remark}
\newtheorem{rem}{Remark}[thm]
\theoremstyle{definition}
\newtheorem{defn}[thm]{Definition}

\newcommand{\supp}{\text{supp}}

\theoremstyle{definition}
\newtheorem*{rem*}{Remark}

\begin{document}
\onehalfspacing
%
\title{Recovery Analysis for Weighted $\ell_1$-Minimization Using a Null Space Property} 
\author{
 Hassan Mansour\thanks{Hassan Mansour is with
Mitsubishi Electric Research Laboratories, Cambridge, MA 02139 (mansour@merl.com).}, Rayan Saab\thanks{Rayan Saab is with the Department of Mathematics, The University of California, San Diego (rsaab@ucsd.edu).} 
  }

\maketitle
\begin{abstract}
We study the recovery of sparse signals from underdetermined linear measurements when a potentially erroneous support estimate is available. Our results are twofold. First, we derive necessary and sufficient conditions for signal recovery from compressively sampled measurements using weighted $\ell_1$-norm minimization. These conditions, which depend on the choice of weights as well as the size and accuracy of the support estimate, are on the null space of the measurement matrix. They can guarantee recovery even when standard $\ell_1$ minimization fails. Second, we derive bounds on the number of Gaussian measurements for these conditions to be satisfied, i.e., for weighted $\ell_1$ minimization to successfully recover all sparse signals whose support has been estimated sufficiently accurately. Our bounds show that weighted $\ell_1$ minimization requires significantly fewer measurements than standard $\ell_1$ minimization when the support estimate is relatively accurate.

\end{abstract}

\section{Introduction}

The application of $\ell_1$ norm minimization for the recovery of sparse signals from incomplete measurements has become standard practice since the advent of  compressed sensing \cite{Donoho2006_CS, CRT05, CRT06}. Consider an arbitrary $k$-sparse signal $x \in \R^N$ and its corresponding linear measurements $y \in \R^m$ with $m < N$, where $y$ results from the underdetermined system  
\begin{equation}\label{eq:yAx}
	y = Ax.
\end{equation} 
It is possible to exactly recover all such sparse $x$ from $y$ by solving the $\ell_1$ minimization problem
\begin{equation}\label{eq:L1_min}
	\min_{z}\ \|z\|_{1}\ \text{subject to}\ y = Az
\end{equation}
 if $A$ satisfies certain conditions  \cite{Donoho2006_CS, CRT05, CRT06}. In particular, these conditions are satisfied with high probability by many classes of random matrices, including those whose entries are i.i.d. Gaussian random variables, when\footnotemark[3]\footnotetext[3]{We write $x\gtrsim y$ when $x \geq Cy $ for some constant independent of $x$ and $y$.} $m \gtrsim k\log(N/k)$.

One property of the measurement matrix $A$ that characterizes sparse recovery from compressive measurements is the null space property (NSP) (see, e.g., \cite{cohen2006csa} ) defined below.
\begin{defn}{\cite{cohen2006csa}}\label{def:NSP}
A matrix $A \in \R^{m\times N}$ is said to have the null space property of order $k$ and constant $C$ if for any vector $h:Ah=0$, and for every index set $T \subset \{1\dots N\}$ with $|T| \leq k$, we have 
\[  \|h_T\|_1 \leq C \|h_{T^c}\|_1. \]
In this case, we say that $A$ satisfies NSP$(k,C)$.
\end{defn}
NSP$(k, C)$ with $C < 1$ is a necessary and sufficient condition on the matrix $A$ for the recovery of all $k$-sparse vectors from their measurements using \eqref{eq:L1_min}. For example, it was shown in \cite[Section 9.4]{FoucartRauhut:2013} using an \emph{escape through the mesh argument}\cite{EscapeMesh:1988} based on, e.g.,\cite{Rudelson07} (cf. \cite{Stojnic:2009,CRPW:2012}) that Gaussian random matrices satisfy the null space property with probability greater than $1-\epsilon$ when $m > c k\ln{eN/k}$. Here, $c$ depends on $C$ and $\epsilon$, but the dependence is mild enough that $c\approx 8$ is a reasonable approximation when $N$ is large and $k/N$ is small. 

While the $\ell_1$ minimization problem in \eqref{eq:L1_min} is suitable for recovering signals with arbitrary support sets, it is often the case in practice that signals exhibit structured support sets, or that an estimate of the support can be identified. In such cases, one is interested in modifying  \eqref{eq:L1_min} to weaken the exact recovery conditions. In this paper, we analyze a recovery method that incorporates support information by replacing \eqref{eq:L1_min} with weighted $\ell_1$ minimization. In particular, given a support estimate $\widetilde{T}\subset\{1,...,N\}$, we solve the optimization problem
\begin{equation}\label{eq:wL1_min}
  \min_{z} \sum_{i=1}^N w_i |z_i| \ \text{subject to}\ y = Az, \textrm{ where } w_i = \left\{ \begin{array}{cc}  \w \in [0,1], & i\in \widetilde{T} \\ 1, & i \in \widetilde{T}^c\end{array}\right. .
\end{equation}

The idea behind such a modification is to choose the weight vector such that the entries of $x$ that are ``expected'' to be large, i.e., those on the support estimate $\widetilde{T}$, are penalized less.
\subsection{Prior work}
The recovery of compressively sampled signals using prior support
information has been studied in several works, e.g., \cite{CS_using_PI_Borries:2007, Vaswani_Lu_Modified-CS:2010,
  Vaswani_ISIT:2009, reg_mod_BPDN_Vaswani:2010, Jacques:2010,
  WL1_min_Hassibi:2009, FMSY:2012, Mansour:SPIE11,OKH:2012, ISDWangYin:2010}. Vaswani and Lu \cite{Vaswani_Lu_Modified-CS:2010,
  Vaswani_ISIT:2009, reg_mod_BPDN_Vaswani:2010} proposed a modified
compressed sensing approach that incorporates known support
elements using a weighted $\ell_1$ minimization approach, with zero
weights on the known support. Their work derives sufficient recovery
conditions that are weaker than the analogous $\ell_1$ minimization conditions of
\cite{CRT05} in the case where a large proportion of the support is
known. This work was
extended by Jacques in \cite{Jacques:2010} to include
compressible signals and account for noisy measurements.

Friedlander et al. \cite{FMSY:2012} studied the case where non-zero weights are applied to the support estimate, further generalizing and refining the results of Vaswani and Liu; they derive tighter sufficient recovery conditions that depend on the accuracy and size of the support estimate. Mansour et al. \cite{Mansour:SPIE11} then extended these results to incorporate multiple support estimates with varying accuracies. 

Khajehnejad et al.~\cite{WL1_min_Hassibi:2009} also derive sufficient recovery conditions for compressively sampled signals with prior support information using weighted $\ell_1$ minimization. They partition $\{1,\ldots,N\}$ to two sets such that the entries of $x$ supported on each set have a fixed probability of being non-zero, albeit the probabilities differ between the sets. Thus, in this work the prior information is the knowledge of the partition and probabilities. More recently, Oymak et al. \cite{OKH:2012} adopt the same prior information setup as \cite{WL1_min_Hassibi:2009} and derive lower bounds on the the minimum number of Gaussian measurements required for successful recovery when the optimal weights are chosen for each set. Their results are asymptotic in nature and pertain to the non-uniform model where one fixes a signal and draws the matrix at random. In this model, every new instance of the problem requires a new draw of the random measurement matrix. In addition to differing in our model for prior information, our results are uniform in nature, i.e., they pertain to the model where the matrix is drawn once and successful recovery is guaranteed (with high probability) for all sparse signals with sufficiently accurate support information.

Recently, Rauhut and Ward~\cite{RW_interpoation_wL1:2014} analyzed the effectiveness of weighted $\ell_1$ minimization for the interpolation of smooth signals that also admit a sparse representation in an appropriate transform domain. Using a weighted robust null space property, they derive error bounds associated with recovering  functions with coefficient sequences in weighted $\ell_p$ spaces, $0 < p < 1$. This differs from our work which focuses on the effect of the support estimate accuracy on the recovery guarantees, and the relationship between the weighted null space property and the standard null space property (e.g.,\cite{cohen2006csa}).

\subsection{Notation and preliminaries}
Throughout the paper, we adopt the following notation. As stated earlier, $x\in\R^N$ is the $k$-sparse signal to be recovered and $y\in\R^m$ denotes the vector of measurements, i.e., $y = Ax$.  Thus, $k, N,$ and $m$ denote the number of non-zeros entries of $x$, its ambient dimension, and the number of measurements, respectively. $T\subset\{1,...,N\}$ is the support of $x$, and $\widetilde{T}$ is the support estimate used in \eqref{eq:wL1_min}. The cardinality of $\widetilde{T}$ is $|\widetilde{T}|=\rho k$ for some $\rho > 0$ and the accuracy of $\widetilde{T}$ is $\alpha = \frac{|\widetilde{T} \cap T|}{|\widetilde{T}|}$. For an index set ${V}\subset\{1,...,N\}$ we define $${\Gamma}_s(V):=\left\{ U \subset \{1,...,N\} : \big|(V\cap U^c) \cup (V^c\cap U)\big|\leq s\right\}.$$
We use the notation $x_T$ to denote the restriction of the vector $x$ to the set $T$. 
We introduce a \textit{weighted nonuniform} null space property that, as we prove in Section \ref{sec:wNSP}, provides a necessary and sufficient condition for the recovery of sparse vectors supported on a fixed set using weighted $\ell_1$ minimization (with constant weights applied to a support estimate).
\begin{defn}\label{def:wNSP}
Let $T \subset \{1\dots N\}$ with $|T| \leq k$ and $\widetilde{T}\in \Gamma_s{\left(T\right)}$. A matrix $A \in \R^{m\times N}$ is said to have the weighted nonuniform null space property with parameters $T$ and $\widetilde{T}$, and constant $C$ if for any vector $h:Ah=0$, we have 
\[  \w \|h_T\|_1 + (1-\w)\|h_S\|_1 \leq C \|h_{T^c}\|_1, \]
where $S = (\widetilde{T}\cap T^c) \cup (\widetilde{T}^c\cap T)$. In this case, we say $A$ satisfies $\w$-NSP$(T,\widetilde{T},C)$.
\end{defn}

Next we define a \textit{weighted uniform} null space property that lends itself to necessary and sufficient conditions for the recovery of all $k$-sparse vectors from compressive measurements using weighted $\ell_1$ minimization.

\begin{defn}\label{def:wNSP}
A matrix $A \in \R^{m\times N}$ is said to have the weighted null space property with parameters $k$ and $s$, and constant $C$ if for any vector $h:Ah=0$, and for every index set $T \subset \{1\dots N\}$ with $|T| \leq k$ and $S \subset \{1\dots N\}$ with $|S| \leq s$, we have 
\[  \w \|h_T\|_1 + (1-\w)\|h_S\|_1 \leq C \|h_{T^c}\|_1. \]
In this case, we say $A$ satisfies $\w$-NSP$(k,s,C)$.
\end{defn}
Thus, the standard null space property of order $k$, i.e.,  NSP$(k,C),$ is equivalent to $1$-NSP$(k,k,C)$. 
\begin{rem}There should be no confusion between the notation used for the weighted non-uniform and uniform null space properties, as one pertains to subsets and the other to sizes of subsets. 
\end{rem}

\subsection{Main contributions}

\noindent{\bf Necessary and sufficient conditions:} Our first main result is Theorem \ref{thm:NecSuff_wNSP}, identifying necessary and sufficient conditions for weighted $\ell_1$ minimization to recover all $k$-sparse vectors when the error in the support estimate is of size $s$ or less. 

\begin{thm}\label{thm:NecSuff_wNSP}

Given a matrix $A\in \R^{m\times N}$, every $k$-sparse vector $x\in \R^N$ is the unique solution to all optimization problems \eqref{eq:wL1_min} with $\widetilde{T}\in \Gamma_s{\big(\supp{(x)}\big)}$ if and only if $A$
satisfies $\w$-NSP$(k,s,C)$ for some $C<1$.
\end{thm}

We prove this theorem in Section \ref{sec:wNSP}. There, we also compare $\ell_1$ minimization to weighted $\ell_1$ minimization. For example, we show that if the accuracy of the support estimate $\widetilde{T}$ is at least $50\%$, then weighted $\ell_1$ minimization recovers $x$ if $\ell_1$ minimization recovers $x$ (Corollary \ref{for:equivalence}). Moreover, when the support accuracy exceeds $50\%$ and the weights are sufficiently small, weighted $\ell_1$ minimization can successfully recover $x$ even when standard $\ell_1$ minimization fails (Corollary \ref{cor:NSPimpliesWeighted2}).

\noindent{\bf Weaker conditions on the number of measurements:} Our second main result deals with matrices $A\in\R^{m \times N}$ whose entries are i.i.d. Gaussian random variables. We establish a condition on the number of measurements, $m$, that yields the weighted null space property, and hence guarantees exact sparse recovery using weighted $\ell_1$ minimization. 
\begin{thm}\label{thm:gauss_nonunif}
Let $T$ and $\widetilde{T}$ be two subsets of $\{1,...,N\}$ with $|T|\leq k$ and $|\Sset| \leq s$ and let $A$ be a random matrix with independent zero-mean unit-variance Gaussian entries. Then $A$ satisfies $\w$-NSP$(T,\widetilde{T},C)$ with probability exceeding $1-\epsilon$ provided 
\[\frac{m}{\sqrt{m+1}} \geq  \sqrt{k+s}+ C^{-1}\sqrt{2(\w^2k + s)\ln(eN/k)}  + \Big(\frac{1}{2\pi e^3}\Big)^{1/4} \sqrt{\frac{k}{\ln(eN/k)}} + \sqrt{2\ln\epsilon^{-1}}.\]
\end{thm}
We observe that in the limiting case of large $m, N, k$ with small $k/N$, and taking $\w = 0$, the condition in Theorem \ref{thm:gauss_nonunif} simplifies to 
\[ m \gtrsim   k + {s\ln(eN/k)},\]
which can be significantly smaller than the analogous condition $m \gtrsim k\ln{(eN/k)}$ of standard $\ell_1$ minimization \cite{FoucartRauhut:2013} especially when the support estimate is accurate, i.e., when $s$ is small. 
In Section \ref{sec:GaussMat}, we prove a more general version of Theorem \ref{thm:gauss_nonunif}, namely Theorem \ref{thm:gauss}. Theorem \ref{thm:gauss} suggests  that the choice $\w = 1 - \alpha$ gives the weakest condition on the number of measurements $m$. On the other hand, Proposition \ref{prop:wNSPimpliesvNSP} shows that when ${\alpha} > 0.5$, and the weighted null space property holds for a weight $\w$, it also holds for weights $\v < \w$. Taken together, these results indicate that while the bound on the number of measurements $m$ is minimized for $\w = 1 - \alpha$, recovery by weighted $\ell_1$ minimization is also guaranteed for all weights $\v \in [0,1 - \alpha]$ (when $\alpha > 0.5$). We note that Theorem \ref{thm:gauss} also indicates that when $\alpha > 0.5$, then using any weight $\w < 1$ results in a weaker condition on the number of measurements $m$ than standard $\ell_1$ minimization. In Section \ref{sec:GaussMat}, we also develop the corresponding bounds that guarantee uniform recovery for arbitrary sets $T$ and $\widetilde{T} \subset \{1\dots N\}$. Finally, we present numerical simulations in Section \ref{sec:results} that illustrate our theoretical results.


\section{Weighted null space property}\label{sec:wNSP}

In what follows, we describe the relationship between the weighted and standard null space properties and their associated optimization problems. Specifically, Proposition \ref{prop:WNSPisNecessary} establishes $\w$-$\operatorname{NSP}(k,s,C)$ with some $C<1$ as  a necessary condition for weighted $\ell_1$-minimization to  recover all $k$-sparse vectors $x$ from their measurements $Ax$, given a support estimate with at most $s$ errors.
Proposition \ref{prop:WNSPisSufficient} establishes that the same weighted null space property is also sufficient. Together, Propositions \ref{prop:WNSPisNecessary} and \ref{prop:WNSPisSufficient} yield Theorem \ref{thm:NecSuff_wNSP}.

Proposition \ref{prop:wNSP_from_NSP} relates the weighted null space property to the standard null space properties of size $s$, $k-s$, and $k$. As a consequence, Corollary \ref{cor:NSPimpliesWeighted2} shows that weighted $\ell_1$ minimization can succeed when $\ell_1$ minimization fails provided the support estimate is accurate enough and the weights are small enough. Proposition \ref{prop:wNSPimpliesvNSP} shows that if $s < k$, i.e. the support estimate is at least $50\%$ accurate, then any matrix that satisfies $\w$-NSP$(k,s,C)$ also satisfies $\v$-NSP$(k,s,C)$ for all $\v < \w$. Corollary \ref{cor:NSPimpliesWeighted} shows that the standard null space property guarantees that weighted $\ell_1$ minimization succeeds when the support estimate is at least $50\%$ accurate, regardless of $\w\in (0,1] $. Corollary \ref{for:equivalence} establishes the equivalence of $\w$-NSP$(s,s,C_s)$ and NSP$(s,C_s)$. This shows that weighted $\ell_1$ minimization succeeds in recovering all $s$-sparse signals from a support estimate that is $50\%$ accurate if and only if $\ell_1$ minimization recovers alls $s$-sparse signals.

\begin{prop}\label{prop:WNSPisNecessary}
Let $A$ be an $m\times n $ matrix that does not satisfy $\w$-NSP$(k,s,C)$ for any $C<1$. Then, there exists a $k$-sparse vector $x$ satisfying $Ax = b$ and a set $\widetilde{T}$ with 
$|({\widetilde{T}}\cap T^c) \cup ({\widetilde{T}^c}\cap T)|\leq s$ such that $x$ is not the unique minimizer of the optimization problem \eqref{eq:wL1_min}.
\end{prop}

\begin{proof}
Since $A$ does not satisfy $\w$-NSP$(k,s,C)$ for any $C<1$, there exists a vector $h: Ah=0$ and  index sets $T$ with $|T|\leq k$ and $S$ with $|S|\leq s$ such that 
$Ah_T=-Ah_{T^c}$ and $$\w\|h_T\|_1 + (1-\w)\|h_S\|_1 \geq \|h_{T^c}\|_1.$$
Define $\widetilde{T}:= (T^c \cap S)\cup(T\cap S^c)$, so that  $S= (T\cap\widetilde{T}^c) \cup (T^c \cap \widetilde{T})$. Substituting for $S$, splitting the set $T$, and simplifying we obtain
\[ \| h_{{T\cap\widetilde{T}^c}} \|_1 + \w \| h_{{T\cap\widetilde{T}}} \|_1    \geq \w\| h_{T^c\cap\widetilde{T}} \|_1+  \| h_{T^c\cap\widetilde{T}^c} \|_1.\] 
In other words, the weighted $\ell_1$-norm of the vector $h_T$ equals or exceeds that of $h_{T^c}$. So  $h_T$ is not the unique minimizer of \eqref{eq:wL1_min}.
This establishes the necessity of the $\w$-NSP condition. 
\end{proof}

\begin{prop}\label{prop:WNSPisSufficient}
Let $A$ be an $m\times n $ matrix that satisfies $\w$-NSP$(k,s,C)$ for some $C<1$. Then, every  $k$-sparse vector $x$ is the unique minimizer of the optimization problem \eqref{eq:wL1_min} provided $\widetilde{T}$ satisfies $|({\widetilde{T}}\cap T^c) \cup ({\widetilde{T}^c}\cap T)|\leq s$.
\end{prop}
\begin{proof}
Let $x^*$ be a minimizer of \eqref{eq:wL1_min} and define $h:=x^*-x$. Then by the optimality of $x^*$, and using the reverse triangle inequality \[  \w \|h_{T^c}  \|_1 + (1-\w)\|h_{{\widetilde{T}^c}\cap T^c}  \|_1 \leq \w \|h_{T}  \|_1 + (1-\w)\|h_{\widetilde{T}^c \cap T}  \|_1. 
\]
Consequently
\[  \|h_{T^c}  \|_1  \leq \w \|h_{T}  \|_1 +  (1-\w)\|h_{{\widetilde{T}}\cap T^c}  \|_1 + (1-\w)\|h_{\widetilde{T}^c \cap T}  \|_1.\]
Setting $S=({\widetilde{T}}\cap T^c) \cup ({\widetilde{T}^c}\cap T)$, we note that when $|S|\leq s$, the above inequality is in contradiction with $\w$-NSP$(k,s,C)$ for $C<1$ unless $h=0$. We thus conclude that $x^*=x$.
\end{proof}

\begin{prop}\label{prop:wNSP_from_NSP}
Let $A$ be an $m\times n $ matrix that satisfies $1$-NSP$(s,s,C_s)$ for some $C_s<1$ as well as $1$-NSP$(k-s,k-s,C_{k-s})$ and $1$-NSP$(k,k,C_k)$ for some finite $C_{k-s},C_k$. Then, $A$ satisfies $\w$-NSP$(k,s,C(\w))$, with \quad $C(\w)=\frac{(1+\w)C_sC_{k-s} +C_s+\w C_{k-s}}{1-C_s C_{k-s}}$. 
\end{prop}
\begin{proof}
Let $h$ be any fixed vector in the null space of $A$ and let $T^*$ and $S^*$ be the supports of its $k$-largest and $s$-largest entries (in modulus), respectively.  To check whether $\w\|h_T\|_1+(1-\w)\|h_S\|_1 \leq C \|h_{T^c}\|_1$ holds for all $T:|T|\leq k$ and $S: |S|\leq s$, it suffices to check whether $\w\|h_{T^*}\|_1+(1-\w)\|h_{S^*}\|_1 \leq C \|h_{{(T^*)}^c}\|_1$ holds. To see this, note that the left hand side is largest and the right hand side is smallest over all choices of $S:|S| \leq s$ and $T:|T|\leq k$ when $(T,S)=(T^*,S^*)$. 
Since $A$ satisfies $1$-NSP$(s,s,C_s)$ and $1$-NSP$(k,k,C_k)$, $h_{T^*}$ and $h_{S^*}$ have no zero entries on $T^*$ or $S^*$, respectively (otherwise we would have $h_{(T^*)^c}=0$, contradicting the null space property). 
Thus, to prove $\w$-NSP$(k,s,C)$ we may now examine for any vector $h: Ah=0$ only sets $T$, $S$ with $S\subset T$ and $|T|=k$, $|S|=s$. Defining $\widetilde{T}:=T\setminus S$ (which implies $S=T\setminus{\widetilde T}$, and $|\widetilde{T}|=k-s$) we have
\begin{align}
\w\|h_T\|_1+(1-\w)\|h_S\|_1 &= \w\|h_{\widetilde{T}}\|_1 + \|h_{T\cap\widetilde{T}^c}\|_1 \nonumber \\
&\leq \w \ C_{k-s}( \|h_{T\cap\widetilde{T}^c}\|_1+\|h_{T^c}\|_1)+\|h_{T\cap \widetilde{T}^c}\|_1.
\label{eq:some_eq}\end{align}
Moreover,  we have  $\|h_{\widetilde{T}}\|_1 \leq C_{k-s} (\|h_{T\cap{\widetilde T}^c}\|_1+\|h_{T^c}\|_1)$. Hence,
\begin{align}
\|h_{T\cap\widetilde{T}^c}\|_1 &\leq C_s (\|h_{T^c}\|_1+\|h_{\widetilde{T}}\|_1)\nonumber\\
 &\leq C_s \Big((1+C_{k-s})\|h_{T^c}\|_1+ C_{k-s}\|h_{T\cap\widetilde{T}^c}\|_1\Big) \nonumber, 
\end{align}
and consequently $  \|h_{T\cap\widetilde{T}^c}\|_1 \leq  \frac{C_s (1+C_{k-s})}{1-C_s C_{k-s}} \|h_{T^c}\|_1$. Substituting in \eqref{eq:some_eq}, we obtain
\begin{equation}
\w\|h_T\|_1+(1-\w)\|h_S\|_1 \leq \frac{(1+\w)C_sC_{k-s} +C_s+\w C_{k-s}}{1-C_s C_{k-s}} \|h_{T^c}\|_1, 
\end{equation} which is the desired result.
\end{proof}

\begin{cor}\label{cor:NSPimpliesWeighted2}
Let $k$ be a positive integer  and suppose that $C_{k-s}>1$ is the smallest constant so that the $m\times n$ matrix  $A$ satisfies $1$-NSP$(k-s,k-s,C_{k-s})$. 
Suppose there exists an integer $s<k$ such that $A$ satisfies $1$-NSP$(s,s,C_s)$ with constant $C_s<1/(2C_{k-s}+1)$.
 Let $x$ be any $k$-sparse vector in $\R^n$, with $\textrm{supp}(x) = T$. If $\widetilde{T}$ satisfies 
$|\widetilde{T}\cap T^c| + |\widetilde{T}^c\cap T | \leq s $ and $\w \leq  \frac{1-2C_s C_{k-s}-C_s}{C_{k-s}(C_s+1)}$, then  $x=x^*(\w,\widetilde{T})$, the minimizer of \eqref{eq:wL1_min} with $y=Ax$.
\end{cor}
\begin{proof}
 Proposition \ref{prop:wNSP_from_NSP} implies that $A$ satisfies $\w$-NSP$(k,s,C(\w))$ with $C(\w)=\frac{(1+\w)C_sC_{k-s} +C_s+\w C_{k-s}}{1-C_s C_{k-s}}$. If $\w\leq \frac{1-2C_s C_{k-s}-C_s}{C_{k-s}(C_s+1)}$ and $C_s < \frac{1}{2C_{k-s}+1}$ then $0 \leq C(\w) < 1$, so Proposition \ref{prop:WNSPisSufficient} guarantees that $x=x^*(\w,\widetilde{T})$.
\end{proof}

\begin{prop}\label{prop:wNSPimpliesvNSP}
Let $A$ be an $m\times n$ matrix that satisfies $\w$-NSP$(k,s,C)$. If $s \leq k$, then for every weight $\v \leq \w$, the matrix $A$ satisfies $\v$-NSP$(k,s,C)$. 
\end{prop}
\begin{proof} Since $A$ satisfies $\w$-NSP$(k,s,C)$, then any vector $h : Ah = 0$ satisfies
\[ \w\|h_T\|_1+(1-\w)\|h_S\|_1 \leq C \|h_{T^c}\|_1, \]
for all index sets $T \subset \{1\dots N\}$ with $|T| \leq k$ and $S \subset \{1\dots N\}$ with $|S| \leq s$. In particular, consider the sets $T^*= T^*(h)$ and $S^* = S^*(h)$ indexing the largest $k$ and $s$ entries in magnitude of $h$. We have
\[ \w\|h_{T^*}\|_1+(1-\w)\|h_{S^*}\|_1 \leq C \|h_{T^{*c}}\|_1. \]
Let $\v < \w$ and write the above as
\[ \v\|h_{T^*}\|_1 + (\w-\v)\|h_{T^*}\|_1+(1-\w)\|h_{S^*}\|_1 \leq C \|h_{T^{*c}}\|_1. \]
Since $s < k$, the set\ $S^* \subset T^*$, which implies
\[ \v\|h_{T^*}\|_1 + (\w-\v)\|h_{S^*}\|_1+(1-\w)\|h_{S^*}\|_1 \leq C \|h_{T^{*c}}\|_1, \]
which is equivalent to 
\[ \v\|h_{T^*}\|_1 + (1-\v)\|h_{S^*}\|_1 \leq C \|h_{T^{*c}}\|_1. \]
Replacing $T^*$ by an arbitrary $T$ of the same size and $S^*$ by an arbitrary $S$ of the same size decreases the left hand side. Replacing $T^{*c}$ by $T^c$ increases the right hand side. So,
\[ \v\|h_{T}\|_1 + (1-\v)\|h_{S}\|_1 \leq C \|h_{T^{c}}\|_1, \]
for all $S,T$  with $|S| \leq s$ and $|T| \leq k$. This is $\v$-NSP($k,s,C$), and it holds for all $\v<\w$ once it holds for $\w$. In particular, it holds for $\v=0$. 
\end{proof}
\begin{rem} The condition $s \leq k$ is satisfied when a support estimate set $\widetilde{T}$ with $|({\widetilde{T}}\cap T^c) \cup ({\widetilde{T}^c}\cap T)|\leq s$ has an accuracy $\alpha \geq 0.5$. Therefore, Proposition \ref{prop:wNSPimpliesvNSP} states that if the support estimate is at least $50\%$ accurate, any matrix $A$ that satisfies $\w$-NSP$(k,s,C)$ also satisfies $\v$-NSP$(k,s,C)$ for every weight $\v < \w$.
\end{rem}

\begin{cor}\label{cor:NSPimpliesWeighted}
Let $A$ be an $m\times n$ matrix that satisfies $1$-NSP$(k,k,C_k)$ with $C_k<1$. Then, for every $k$-sparse vector $x$ supported on some set $T$, and for every support estimate $\widetilde{T}$ with $\alpha:= \frac{|T\cap \widetilde{T}|}{|\widetilde{T}|}\geq\frac{1}{2}$ it holds that  $x=x^*(\w,\widetilde{T})$, the minimizer of \eqref{eq:wL1_min} with $b=Ax$, and $0 \leq \w<1$.
\end{cor}
\begin{proof}
This follows from Proposition \ref{prop:wNSPimpliesvNSP} by setting $\w=1$ and $s=k$ (and applying Proposition \ref{prop:WNSPisSufficient}).

\end{proof}

\begin{cor}\label{for:equivalence} The weighted null space property $\w$-NSP$(s,s,C_s)$ and the standard null space property $1$-NSP$(s,s,C_s)$ are equivalent. 
\end{cor}

\begin{proof}
Proposition \ref{prop:wNSP_from_NSP} with $k=s$, coupled with the observation $C_0=0$, yield one direction of the equivalence. The other direction, i.e., that $\w$-NSP$(s,s,C_s)$ implies $1$-NSP$(s,s,C_s)$ follows upon picking $S=T$ for any set $T:|T|\leq s$ in the definition of the weighted null space property. 
\end{proof}

\begin{rem}
Corollary \ref{for:equivalence} in turn implies that weighted $\ell_1$-minimization recovers all $s$-sparse signals $x$ from noise-free measurements $Ax$ given a support estimate that is $50\%$ accurate if and only if $\ell_1$ minimization recovers all $s$-sparse signals from their noise-free measurements. 
\end{rem}


\section{Gaussian Matrices}\label{sec:GaussMat}
It is known that Gaussian (and more generally, sub-Gaussian) matrices satisfy the standard null space property (with high probability) provided $m>Ck\log(n/k)$. It is also known, (see, e.g., \cite[Theorem 10.11]{FoucartRauhut:2013}) that if a matrix $A\in\R^{m\times n}$ guarantees recovery of all $k$-sparse vectors $x$ via $\ell_1$ minimization \eqref{eq:L1_min}, then $m$ must exceed $c_1  k \log(\frac{n}{c_2 k})$ for some appropriate constants $c_1$ and $c_2$. The purpose of this section is to show that weighted $\ell_1$ minimization allows us recover sparse vectors beyond this bound (i.e., with fewer measurements) given relatively accurate support estimates. 

We begin with some simple observations to establish a rough lower bound on the number of measurements needed for weighted $\ell_1$ minimization. We first observe that when $k\geq s$, $\w$-NSP$(k,s,C)$ implies $1$-NSP$(s,s,C)$, i.e., the standard null space property of size $s$. This can be seen by restricting $T$ to be of size at most $s$ and setting $S=T$ in the definition of the weighted null space property. Thus, $\w$-NSP$(k,s,C)$ guarantees recovery of all $s$ sparse signals via $\ell_1$ minimization. Consequently, it requires $m\geq c_1 s \log(\frac{N}{c_2 s})$ \cite[Theorem 10.11]{FoucartRauhut:2013}. Since in weighted $\ell_1$ minimization $s$ plays the role of the size of the error in the support estimate, then one may hope (in analogy with standard compressed sensing results) that $m\approx s\log\frac{N}{s}$ suffices for recovery, given an accurate support estimate. However, even if one had a perfect support estimate, $k$ measurements are needed to directly measure the entries on the support. Combining these observations, we seek a bound on the number of measurements that scales (up to constants) like $k + s\log\frac{N}{s}$.

Indeed this can be deduced from Corollary \ref{cor:main_cor}, presented later in this section, which follows from our main technical result, Theorem~\ref{thm:gauss} (which is a more general version of Theorem \ref{thm:gauss_nonunif} in the Introduction). Corollary \ref{cor:main_cor} entails that in the case of large $m,N,k$ and small $k/N$, for a fixed support estimate $\widetilde{T}$ all $k$-sparse vectors supported on any set $T\in \Gamma_s({\widetilde{T}})$ can be recovered by solving \eqref{eq:wL1_min} when $m\gtrsim k + s\log{N/s}$.  We conclude the section with another corollary of Theorem~\ref{thm:gauss}, establishing a bound on the number of Gaussian measurements that guarantee $\w$-NSP$(k,s,C)$.

\begin{thm}\label{thm:gauss}

Let $T$ and $\widetilde{T}$ be two subsets of $\{1,...,N\}$ with $|T|\leq k$ and $|\Sset| \leq s$ and let $A$ be a random matrix with independent zero-mean unit-variance Gaussian entries. Then $A$ satisfies $\w$-NSP$(T,\widetilde{T},C)$ with probability exceeding $1-\epsilon$ provided 
\begin{align}\frac{m}{\sqrt{m+1}} \geq &  \sqrt{s +  \alpha\rho k}+ C^{-1}\sqrt{2((\w^2 - 2\w(1-\alpha))\rho k+ s)\ln(eN/k)} \nonumber\\  &+ \Big(\frac{1}{2\pi e^3}\Big)^{1/4} \sqrt{\frac{k}{\ln(eN/k)}} + \sqrt{2\ln\epsilon^{-1}}.\label{eq:m_bnd}\end{align}
\end{thm}
\proof Our proof will be a modified version of the analogous proof for the standard null space property for Gaussian matrices \cite{Rudelson07},  cf., \cite{FoucartRauhut:2013}, \cite{CRPW:2012}. Define the set \[H_{T,\widetilde{T}} := \big\{ h \in \R^{N}: \w\|h_{T}\|_1 + (1-\w) \| h_{\Sset}\|_1  \geq C \| h_{T^c} \|_1\big \} .\] 
Our aim is to show that for a random Gaussian matrix $A$ with zero-mean and unit variance entries  $\inf\limits_{h\in H_{T,\widetilde{T}}} \|Ah\|_2>0$ with high probability. This will show that there are no vectors from $H_{T,\widetilde{T}}$ in the null space of $A$, i.e.,  that the weighted null space property holds over $T$ and $\widetilde{T}$. To this end, we will utilize Gordon's escape through the mesh theorem \cite{EscapeMesh:1988}, as in \cite{Rudelson07},  cf., \cite{FoucartRauhut:2013}. In short, the theorem states that for an $m\times N$ Gaussian matrix with zero-mean  and unit-variance entries and for an arbitrary set $V\subset S^{N-1}$, 
\begin{equation}P\big(\inf\limits_{v\in V} \|Av\|_2 \leq \frac{m}{\sqrt{m+1}} - \ell(V) - a\big) \leq e^{-a^2/2},\label{eq:Gordon}\end{equation} where $\ell(V):= \mathbb{E}_{g\sim \mathcal{N}(0,I_N)} \sup\limits_{v\in V} \langle  g, v\rangle$ is the Gaussian width of $V$, cf. \cite{CRPW:2012}. So we must estimate $\ell(H_{T,\widetilde{T}}\cap S^{N-1})$. 
Note that $H_{T,\widetilde{T}}\cap S^{N-1}$ is compact so the supremum in the definition of $\ell(H_{T,\widetilde{T}}\cap S^{N-1})$ can be replaced by a maximum. Moreover, note that
  \[ \max_{h\in H_{T,\widetilde{T}}\cap S^{N-1}}  \langle g, h\rangle = \max_{h \in H_{T,\widetilde{T}}\cap S^{N-1}\cap\{h:\ h_i\geq 0\}}~\sum_{i=1}^{N}  |g_i| h_i .\]
Define the vector $\tilde{g}$ with entries $\tilde{g}_i = |g_i|$, $i\in\{1,...,N\}$, the convex cone $\widetilde{H}_{T,\widetilde{T}}=H_{T,\widetilde{T}} \cap \{h\in \R^N: h_i \geq 0\}$, and its dual $\widetilde{H}^*_{T,\widetilde{T}}:=\{ z\in \R^N: \langle z,h\rangle \geq 0  \textrm{ for all } h \in \widetilde{H}_{T,\widetilde{T}}\}$. We may now use duality (see, e.g., \cite{FoucartRauhut:2013}[B.40]) to conclude that 
  \[ \max_{h\in \widetilde H_{T,\widetilde{T}}}  \langle \tilde{g}, h\rangle \leq \min_{z\in \widetilde H^*_{T,\widetilde{T}}} \| \tilde{g} + z\|_2.\]
  To bound the right hand side from above, we introduce for $t\geq0$ (to be determined later) the set 
  \[
  Q^t_{T,\widetilde{T}}:=\{z\in \R^N: 
 \left\{	\begin{array}{l}
  z_i = \w t \textrm{ for } i\in{T \cap \widetilde{T}}, \\
  z_i = t \textrm{ for } i\in{T \cap \widetilde{T}^c}, \\
  z_i = (1 - \w)t \textrm{ for } i\in{T^c \cap \widetilde{T}}, \\
  z_i \geq -C t \textrm{ for } i\in T^c\cap\widetilde{T}^c 
  \end{array}\right.\}
  \] 
  and we observe that for any two vectors $z\in Q^t_{T,\widetilde{T}}$ and $h\in \widetilde{H}_{T,\widetilde{T}}$
\begin{align}
\sum_{i=1}^N z_i h_i & = t\w \sum_{i\in T\cap \widetilde{T}}h_i + t\sum_{i\in T\cap \widetilde{T}^c}h_i + t(1-\w)\sum_{i\in T^c\cap \widetilde{T}}h_i +\sum_{i\in T^c\cap \widetilde{T^c}}z_ih_i\\ 
&=t\w \|h_{T\cap \widetilde{T}}\|_1+t(1-\w)\|h_{\Sset}\|_1 + t\w \|h_{T\cap \widetilde{T}^c}\|_1 +\sum_{i\in T^c\cap \widetilde{T}^c}z_ih_i \\ 
&= t(\w \|h_T\|_1+(1-\w)\|h_{\Sset}\|_1)  +\sum_{i\in T^c\cap \widetilde{T}^c}z_ih_i \\
&\geq t(\w \|h_T\|_1+(1-\w)\|h_{\Sset}\|_1)  -C t \|h_{T^c \cap \widetilde{T}^c}\|_1\\
&\geq t(\w \|h_T\|_1+(1-\w)\|h_{\Sset}\|_1   -C  \|h_{T^c \cap \widetilde{T}^c}\|_1)\geq 0. 
\end{align} 
Hence $Q^t_{T,\widetilde{T}} \subset \widetilde{H}^*_{T,\widetilde{T}}$ and so for any $t\geq 0$
  \[ \max_{h\in H_{T,\widetilde{T}}\cap S^{N-1}}  \langle g, h\rangle \leq \min_{z\in \widetilde H^*_{T,\widetilde{T}}} \| \tilde{g} + z\|_2 \leq \min_{z\in  Q^t_{T,\widetilde{T}}} \| \tilde{g} + z\|_2.\]
Taking expectations and defining $S_{C t}:\R \to \R$, the soft-thresholding operator with $S_\lambda(x) = \sgn(x) (\max(|x|-\lambda/2)) $, we have 
\begin{align}\label{eq:gausswidth}
\ell(H_{T,\widetilde{T}}\cap S^{N-1})  \leq &\mathbb{E}\min_{z\in  Q^t_{T,\widetilde{T}}} \| \tilde{g} + z\|_2 \nonumber\\
\leq &\mathbb{E} \| \tilde{g}_{T \cup (T^c\cap\widetilde{T})} + z_{T \cup (T^c\cap\widetilde{T})}\|_2 +  \mathbb{E}\min_{z\in  Q^t_{T,\widetilde{T}}} \| \tilde{g}_{ (T^c\cap\widetilde{T}^c)} + z_{(T^c\cap\widetilde{T}^c)}\|_2 \nonumber\\
\leq  &\mathbb{E}\|\tilde{g}_{T \cup (T^c\cap\widetilde{T})}\|_2 + \big(\sqrt{\w^2 \alpha \rho k+ (1 - \w)^2 (1-\alpha)\rho k + (1 - \alpha\rho)k}\big)t \nonumber\\
&+ \ \mathbb{E}\min_{z_i \geq -C t } \Big(\sum_{i\in T^c\cap \widetilde{T}^c} (\tilde{g}_i + z_i)^2\Big)^{1/2} \nonumber\\
\leq &\sqrt{(1 + \rho - \alpha\rho)k}+ \big(\sqrt{\w^2 \alpha \rho k+ (1 - \w)^2 (1-\alpha)\rho k + (1 - \alpha\rho)k}\big)t  \nonumber\\
& + \ \mathbb{E} \Big(\sum_{i\in T^c\cap \widetilde{T}^c} S_{C t}(\tilde{g}_i)^2\Big)^{1/2}\nonumber\\
 =  &\sqrt{s +  \alpha\rho k}+ \big(\sqrt{(\w^2 - 2\w(1-\alpha))\rho k+ s}\big)t   +  \mathbb{E} \Big(\sum_{i\in T^c\cap \widetilde{T}^c} S_{C t}(\tilde{g}_i)^2\Big)^{1/2}\nonumber\\
\leq& \sqrt{s +  \alpha\rho k}+ \big(\sqrt{(\w^2 - 2\w(1-\alpha))\rho k+ s}\big)t \nonumber   \\  &+  \Big( (N-s - \alpha\rho k) \big(\frac{2}{\pi e}\big)^{1/2} \frac{e^{-(C t)^2/2} }{(C t)^2} \Big)^{1/2}.
\end{align}

Above, the third inequality utilizes the facts that $z_i=\w t$ when $i\in T\cap \widetilde{T} $ where $|T\cap \widetilde{T} | = \alpha\rho k $,  $z_i=(1 - \w)t$ when $i\in (T^c \cap \widetilde{T})$ where $|T^c \cap \widetilde{T}| = (1 - \alpha\rho)k$, and that $z_i=t$ when $i\in (T \cap \widetilde{T}^c)$ where $|(T^c \cap \widetilde{T})| = (1 - \alpha)\rho k$. 
Moreover, the fourth inequality follows from the well known bound on standard Gaussian random vectors $v\in\R^m$, namely $\mathbb{E}\|v\|_2 \leq \sqrt{m}$ and a straightforward computation for the soft-thresholding term. Similarly the fifth inequality results from a direct computation of $\mathbb{E}S_t(v)^2$ where $v$ is a standard Gaussian vector (see \cite[Section 9.2]{FoucartRauhut:2013} for the details). 

Picking $t=C^{-1}\sqrt{2\ln(eN/k)}$, we have
\begin{align*}\ell(H_{T,\widetilde{T}}\cap S^{N-1})  \leq &\sqrt{s +  \alpha\rho k}+ C^{-1}\sqrt{2((\w^2 - 2\w(1-\alpha))\rho k+ s)\ln(eN/k)} \\ & + \Big(\frac{1}{2\pi e^3}\Big)^{1/4} \sqrt{\frac{k}{\ln(eN/k)}}. \end{align*}
We now apply Gordon's escape through the mesh theorem  \eqref{eq:Gordon} to deduce
\[P\big(\inf\limits_{v\in V} \|Av\|_2 \leq \frac{m}{\sqrt{m+1}} - \ell(H_{T,\widetilde{T}}\cap S^{N-1}) - a\big) \leq e^{-a^2/2}.\]
Choosing $a = \sqrt{2\ln\epsilon^{-1}}$, we obtain $\w$-NSP$(T,\widetilde{T}, C)$, with probability exceeding $1-\epsilon$ provided 
\begin{align*}\frac{m}{\sqrt{m+1}} &\geq  \sqrt{s +  \alpha\rho k}+ C^{-1}\sqrt{2((\w^2 - 2\w(1-\alpha))\rho k+ s)\ln(eN/k)} \\& + \Big(\frac{1}{2\pi e^3}\Big)^{1/4} \sqrt{\frac{k}{\ln(eN/k)}} + \sqrt{2\ln\epsilon^{-1}}.\end{align*}

\qed

\begin{rem} Notice that the term 
\[
 	(\w^2 - 2\w(1-\alpha))\rho k+ s = \left\{
	\begin{array}{ll}
	s \equiv (1 + \rho - 2\alpha\rho)k & \textrm{, when } \w = 0 \\
	k & \textrm{, when } \w = 1
	\end{array}
	\right.
\]
Moreover, the inequality $\alpha\rho \leq 1$ holds. 
\end{rem}

\begin{rem}[Choice of weights] The Gaussian width expression in \eqref{eq:gausswidth} allows us to choose the weight $\w$ that minimizes the upper bound on $\ell(H_{T,\widetilde{T}}\cap S^{N-1})$, in turn leading to the smallest bound on the number of measurements -- as can also be seen in \eqref{eq:m_bnd}. A simple calculation shows that this weight is $\w = 1 - \alpha$. In fact, choosing $\w\in(1-2\alpha, 1)$ guarantees that the second term in \eqref{eq:m_bnd} is smaller than the analogous term for standard $\ell_1$-minimization, i.e., $(\w^2 - 2\w(1-\alpha))\rho k+ s<k$. 
Applying Proposition \ref{prop:wNSPimpliesvNSP} shows that when ${\alpha} > 0.5$, recovery can be performed with any $\w \in (0, 1-{\alpha}]$ with this number of measurements.
\end{rem}

\begin{rem} The condition in Theorem \ref{thm:gauss_nonunif}, i.e., 
\[\frac{m}{\sqrt{m+1}} \geq  \sqrt{k+s}+ C^{-1}\sqrt{2(\w^2k + s)\ln(eN/k)}  + \Big(\frac{1}{2\pi e^3}\Big)^{1/4} \sqrt{\frac{k}{\ln(eN/k)}} + \sqrt{2\ln\epsilon^{-1}},\]
 is a stronger version of the analogous condition in Theorem \ref{thm:gauss}, simplifying the dependence of $m$ on $s$, the size of error in the support estimate. It is obtained by replacing $\alpha\rho$ by 1 and $(\w^2 - 2\w(1-\alpha))\rho k$ by $\w^2 k$ (as these are both upper bounds). 
\end{rem}

\begin{cor}\label{cor:main_cor}
Let  $\widetilde{T}$ be a subset of $\{1,...,N\}$ and let $A$ be a random matrix with independent zero-mean unit-variance Gaussian entries. Then,  with probability exceeding $1-\epsilon$, $A$ satisfies $\w$-NSP$(T,\widetilde{T},C)$ for all sets $T\subset\{1,..,,N\}$ with $|T|\leq k\leq N/2$ and $|\Sset| \leq s\leq k$ provided 
\begin{align}\label{eq:bnd_on_m}\frac{m}{\sqrt{m+1}} \geq  &\Big(1+\frac{1}{(2\pi e^3)^{1/4}\sqrt{\ln(eN/k)}}\Big)\sqrt{k+s}+ C^{-1}\sqrt{2(\w^2 k+s)\ln(eN/k)} \nonumber\\&  + \sqrt{2\ln\epsilon^{-1}+(s+1)\ln(eN/s)+k}.\end{align}
\end{cor}

\begin{proof}
The proof consists of bounding, from above, the number of sets $T$ satisfying the hypotheses of the corollary, and applying Theorem \ref{thm:gauss} in conjunction with a union bound. Denote $p:= |\widetilde{T}| $, then the number of sets $T$ with $|\Sset| \leq s$ is bounded above by $\sum\limits_{i=0}^{s} {{N-p}\choose{i}}{{p}\choose{k-i}}$. This bound is obtained by letting $|T\cap \widetilde{T}^c|$ range from $0$ to $s$ and some simple observations. First, there are ${{N-p}\choose{i}}$ ways to arrange the $i$ entries of $T\cap\widetilde{T}^c$ in $\widetilde{T}^c$. Second, there remain $k-i$ entries of $T$, to be ``placed" in $\widetilde{T}$ and there are ${{p}\choose{k-i}}$ ways to do this. In addition, $p\leq 2k$ as $|T\cap \widetilde{T}|\leq k$ and $|T^c \cap \widetilde{T}| \leq s \leq k$, so 
combining these observations we obtain \[\sum\limits_{i=0}^{s} {{N-p}\choose{i}}{{p}\choose{k-i}} \leq {{N-p}\choose{s}}\sum\limits_{i=0}^{s}{p\choose {k-i}}\leq  (s+1){{N-p}\choose{s}}{{2k}\choose {k}}.\]
Specifically, the last inequality is due to the observation that $p\leq 2k$ and the bound $ {{2k}\choose{k-i}}\leq{{2k}\choose{k}}$.

Thus, applying Theorem \ref{thm:gauss} with $\frac{\epsilon}{(s+1) {{N-p}\choose{s}}{{2k}\choose {k}} }$ in place of $\epsilon$, and a union bound, we conclude that the probability that at least one set $T$ with $|\Sset|\leq s$ violates $\w$-NSP$(T,\widetilde{T}, C)$ is bounded above by $\epsilon$. Note the bound \[\ln\Big( (s+1) {{N-p}\choose{s}}{{2k}\choose {k}} \Big) \leq  (s+1)\ln(eN/s)+k,\] which finishes the proof. 
\end{proof}

\begin{rem}
In the limiting case of large $m,N,k$ with small $k/N$ the condition \eqref{eq:bnd_on_m} simplifies to
\[ m \geq   (\sqrt{k+s}+ C^{-1}\sqrt{2(\w^2 k+s)\ln(eN/k)}   + \sqrt{2\ln\epsilon^{-1}+(s+1)\ln(eN/s)+k})^2,\]
which reveals the benefit of using weighted $\ell_1$-minimization in reducing the number of measurements. In particular, taking $\w=0$ leads to the bound 
\[ m \gtrsim   k + {s\ln(eN/s)},\]
which is essentially as good as one can hope for; one needs $k$ measurements to measure the non-zero entries even if the support was fully known, and about $s\ln(eN/s)$ measurements to recover the entries where the support estimate was erroneous.  
\end{rem}

\begin{rem}
In the case $s=0$ (i.e., perfect support estimation), we have $\widetilde{T}=T$ hence the union bound in the proof above is only over one set, which leads to the better condition (when $\w =0$, and in the limit of large $k,N$) $m>k$. This is to be expected as when the support is known and the weight on the support set $T$ is zero, sparse recovery via weighted $\ell_1$-minimization just requires that the matrix $A_T$ is invertible. In the case of Gaussian matrices, this occurs with probability $1$ if $A_T$ is a square matrix. 
\end{rem}

To simplify the remaining discussion, we introduce one more definition and lemma, which will facilitate the proof of our final result, Corollary \ref{cor:cor_uniform}. 
\begin{defn}
We say that an $m\times n$ matrix $A$ satisfies $\w$-NSP$^*(k,s,C)$ if for all sets $T:|T|=k$ and $S:|S|=s$ with $S\subset T$, and for all vectors $h:Ah=0$ we have $\w \|h_T\|_1 + (1-\w)\|h_S\|_1 \leq C \|h_{T^c}\|_1$.
\end{defn}
\begin{lemma}\label{lem:helper}
$\w$-NSP$(k,s,C)$ is equivalent to $\w$-NSP$^*(k,s,C)$.
\end{lemma}
\begin{proof}
The fact that $\w$-NSP$(k,s,C)$ implies $\w$-NSP$^*(k,s,C)$ follows directly from the definition of the former. To see the reverse implication note that $\w$-NSP$^*(k,s,C)$ implies that for any  $h:Ah=0$ and sets $T^*=T^*(h)$ and $S^*=S^*(h)$ indexing the largest $k$ and $s$ entries (in modulus) of $h$, we have  
$\w \|h_{T^*}\|_1 + (1-\w)\|h_{S^*}\|_1 \leq C \|h_{{T^*}^c}\|_1$. Note that $|T^*|=k$ and $|S^*|=s$, otherwise $h_{{T^*}^c}=0$ and the property fails. Moreover for the vector $h$, $T^*$ and $S^*$ maximize the left-hand side and minimize the right hand side over all choices of $T: |T|\leq k$ and $S:|S|\leq s$.  This argument works for all vectors in the null space so we are done. 
\end{proof}

\begin{cor}\label{cor:cor_uniform}
Let $A$ be a random matrix with independent zero-mean unit-variance Gaussian entries. Then,  with probability exceeding $1-\epsilon$, $A$ satisfies $\w$-NSP$(k,s,C)$ for all sets $T, \widetilde{T}\subset\{1,..,,N\}$ with $|T|\leq k\leq N/2$ and $|\Sset| \leq s\leq k$ provided 
\begin{align}\label{eq:bnd_on_m}\frac{m}{\sqrt{m+1}} \geq  &\Big(1+\frac{1}{(2\pi e^3)^{1/4}\sqrt{\ln(eN/k)}}\Big)\sqrt{k+s}+ C^{-1}\sqrt{2(\w^2 k+s)\ln(eN/k)} \nonumber\\&  + \sqrt{2\ln\epsilon^{-1}+2k\ln(eN/k)}.\end{align}
\end{cor}
\begin{proof}
The proof is similar to the proof of the previous corollary, albeit simpler. By Proposition \ref{lem:helper} we only need to verify $\w$-NSP$^*(k,s,C)$, i.e.,  the null space property pertaining to sets $S:=\Sset \subset T$. There are ${N}\choose{k}$ ways of selecting $T$ and ${k}\choose{s}$ ways of selecting $S\subset T$. We thus apply Theorem \ref{thm:gauss} with $\epsilon$ replaced with $\frac{\epsilon}{ {{N}\choose{k}} {{k}\choose{s}}}$ to complete the proof.
\end{proof}
\begin{rem}
In the limiting case, we see that to obtain $\w$-NSP$(k,s,C)$, we still need a ``baseline"  of $\approx k\ln(eN/k)$ measurements. Nevertheless, this number increases like $C^{-1}(w^2 k+s)\ln{eN/k}$, rather than as $C^{-1}k \ln{eN/k}$ as in the case of $\ell_1$ minimization. Thus, robustness to noise (which improves as $C$ decreases, see, e.g., \cite{FoucartRauhut:2013}) costs fewer measurements in the case of weighted $\ell_1$-minimization, provided the support estimate is relatively accurate. 
\end{rem}

\section{Numerical Examples}\label{sec:results}
In this section, we present phase transition diagrams to illustrate the theoretical results presented above. 
We set $N = 500$ and we draw measurement matrices $A$ with zero mean i.i.d. standard Gaussian random entries of dimensions $m \times N$ where we vary $m$ between 50 and 250 with an increment of 25. We then generate $k$-sparse signals $x\in\R^N$, and vary $k$ between $\lfloor \frac{m}{10} \rfloor$ and $\lfloor \frac{m}{2} \rfloor$. 
The nonzero values in $x$ are drawn independently from a standard Gaussian distribution. 

We generate 50 instances of $A$ and $x$. For each instance, we compute the measurement vector $y = Ax$ and compare the recovery performance using both $\ell_1$ and weighted $\ell_1$ minimization.
Support estimate sets $\widetilde{T}$ of size $k$ with accuracies $\alpha \in \{0.1, 0.3, 0.7, 1\}$ are generated such that $\alpha k$ entries of $\widetilde{T}$ are chosen at random from the support of $x$. The remaining $(1-\alpha)k$ entries of $\widetilde{T}$ are chosen from outside the support of $x$. A weight $\w = 1-\alpha$ is applied to the set $\widetilde{T}$ for every run of weighted $\ell_1$ minimization.

\begin{figure}
\centering
\includegraphics[width=3.5in]{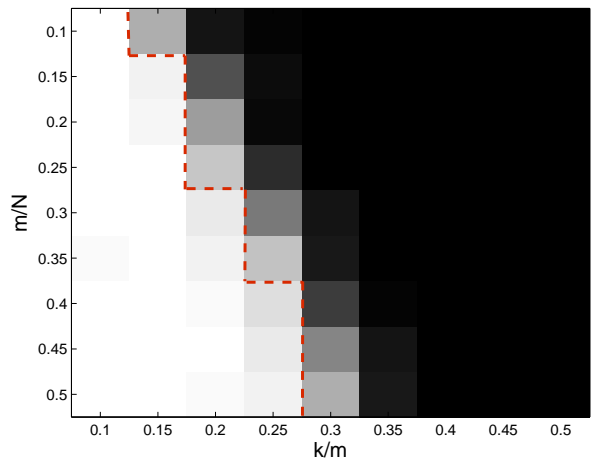}
\caption{Phase transition diagram showing the exact recovery rate using standard $\ell_1$ minimization for 50 instances of the sparse vector $x$ and measurement matrix $A$. The black region indicates a zero recovery rate, while the white region indicates full recovery with rate equal to 1. The dashed red line indicates the empirical recovery rate threshold that exceeds $0.85$.} \label{fig:PT_L1}  
\end{figure}
Fig.~\ref{fig:PT_L1} shows the phase transition diagram illustrating the exact recovery rate over 50 experiments using $\ell_1$ minimization for the recovery of $x$ from measurements $y = Ax$. The recovery rate transitions from a value of 1 (white region) indicating exact recovery for all 50 experiments to a value of zero (black region) indicating errors in the recovery. We highlighted the 0.85 empirical recovery rate threshold by a dashed red line, to compare the recovery performance with weighted $\ell_1$ minimization.

\begin{figure}
\centering
	\hbox{
	\subfigure[$\alpha = 0.1$]{\includegraphics[width=3in]{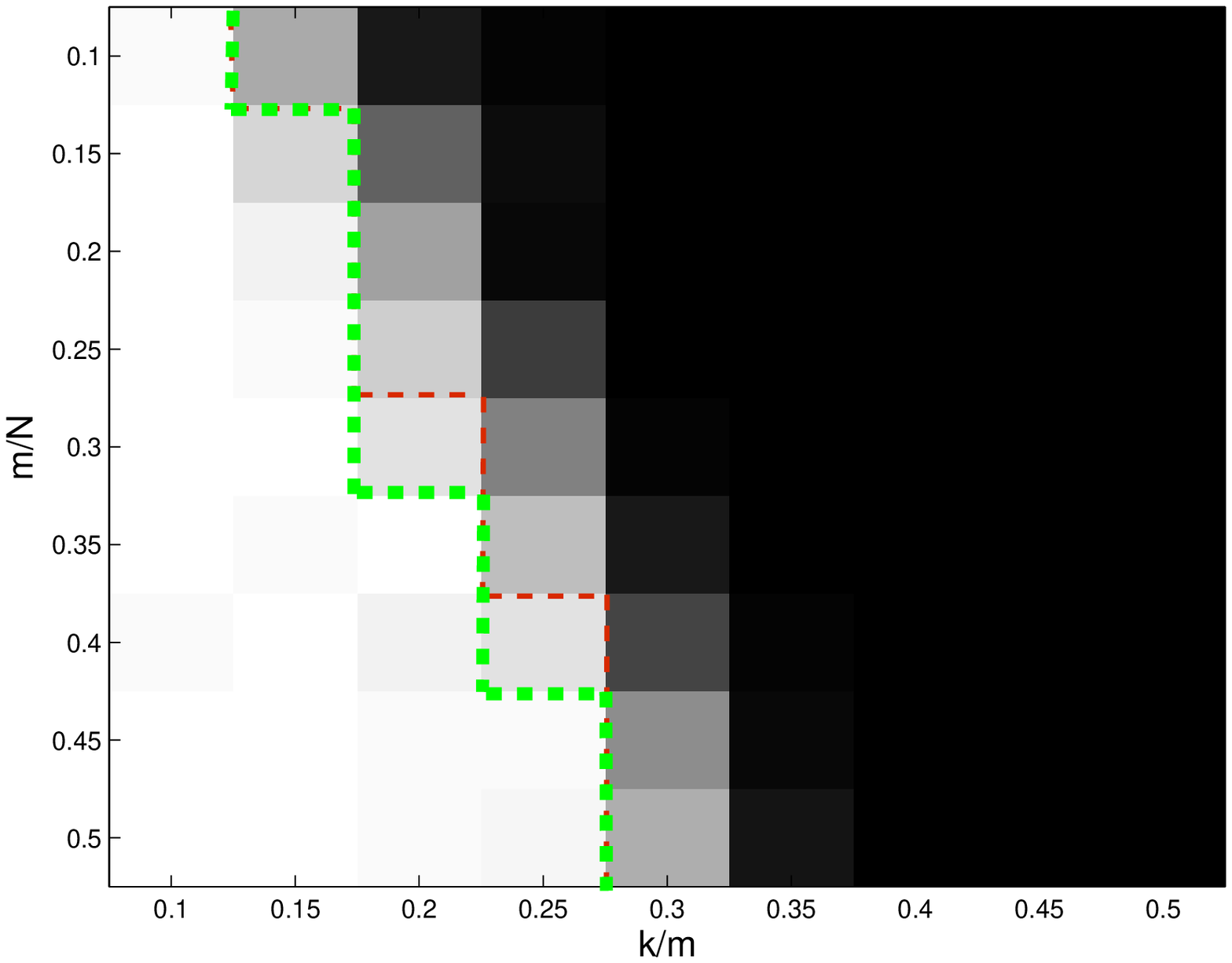}}
	\subfigure[$\alpha = 0.3$]{\includegraphics[width=3in]{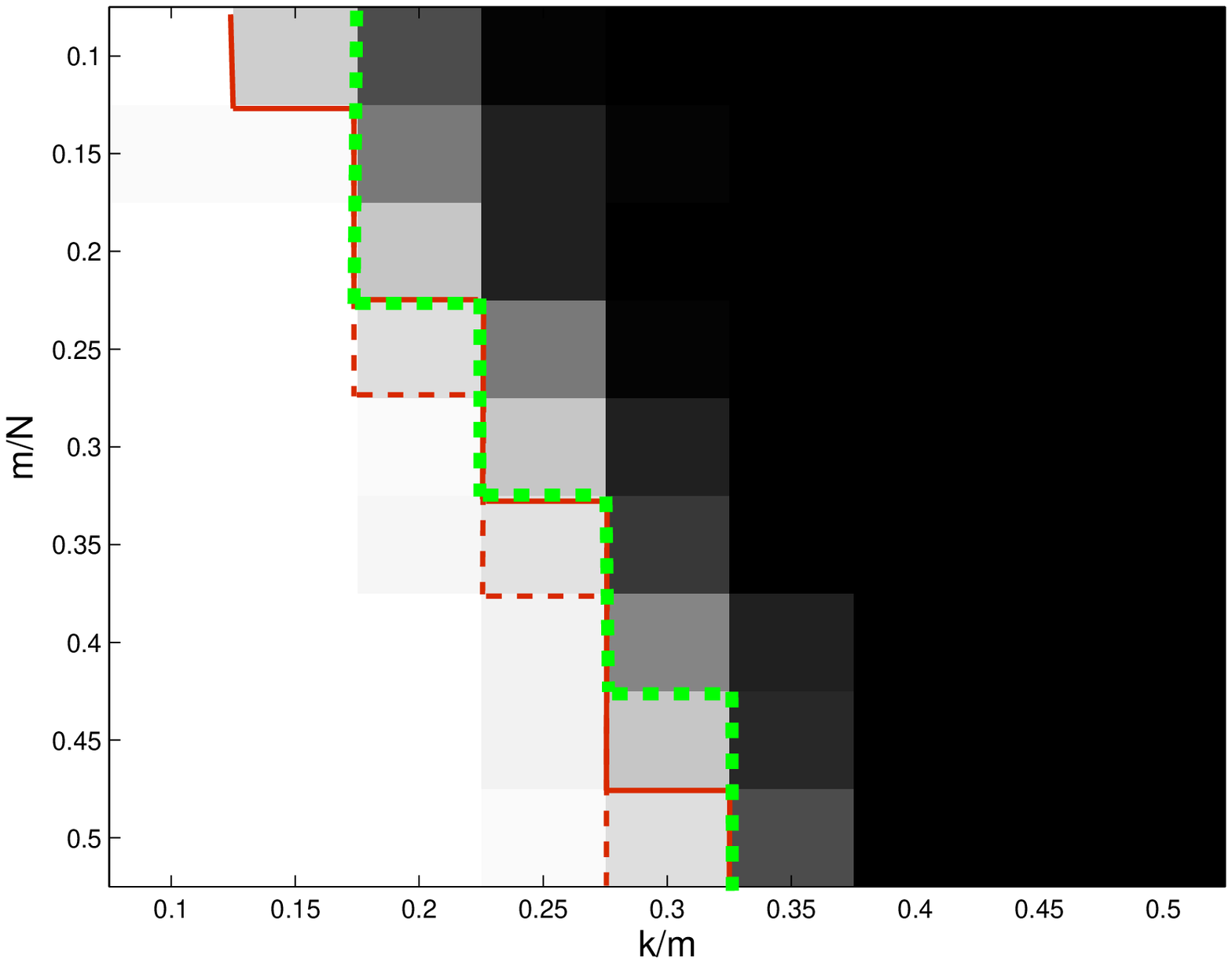}}
	}
	\hbox{
	\subfigure[$\alpha = 0.7$]{\includegraphics[width=3in]{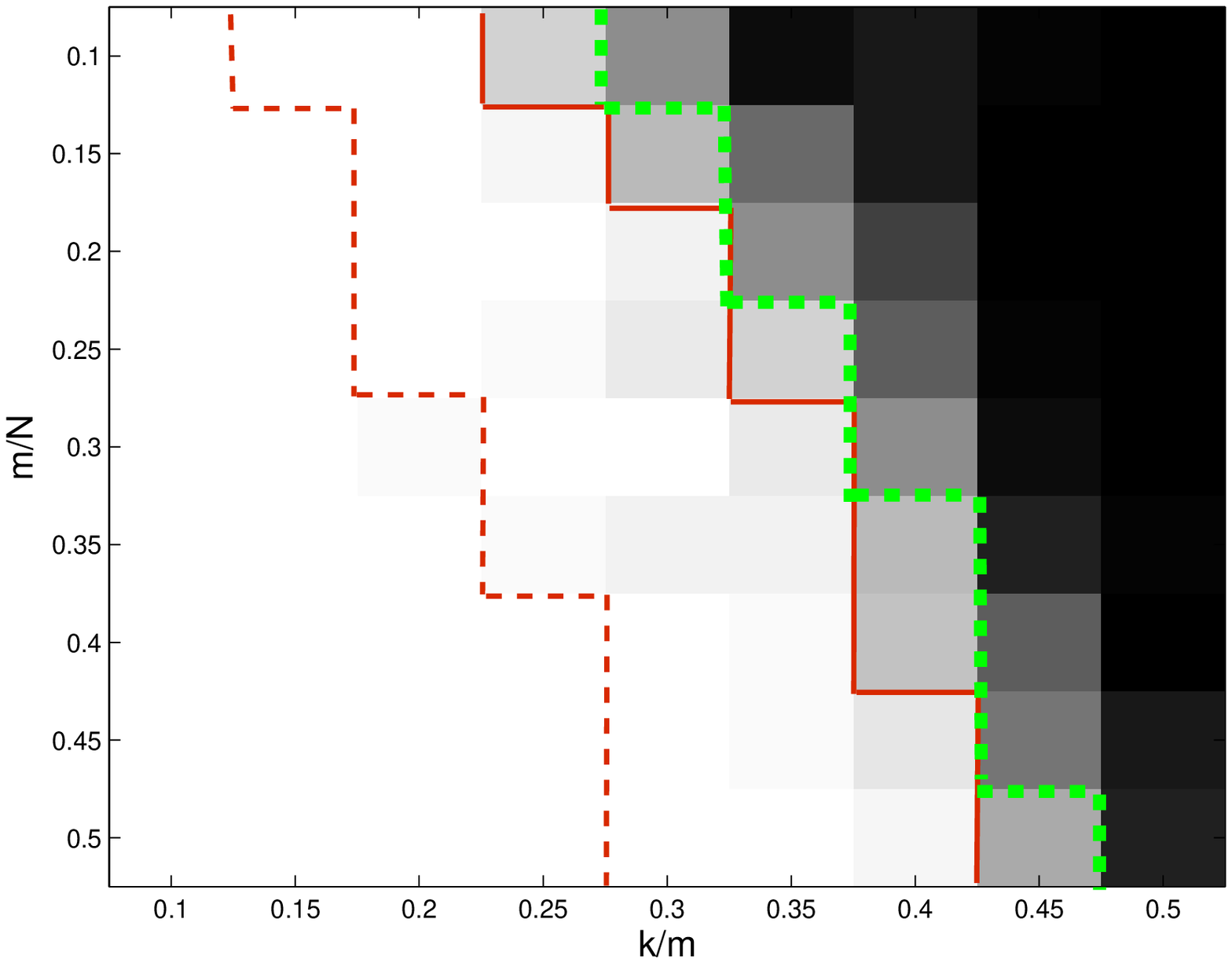}}
	\subfigure[$\alpha = 1$]{\includegraphics[width=3in]{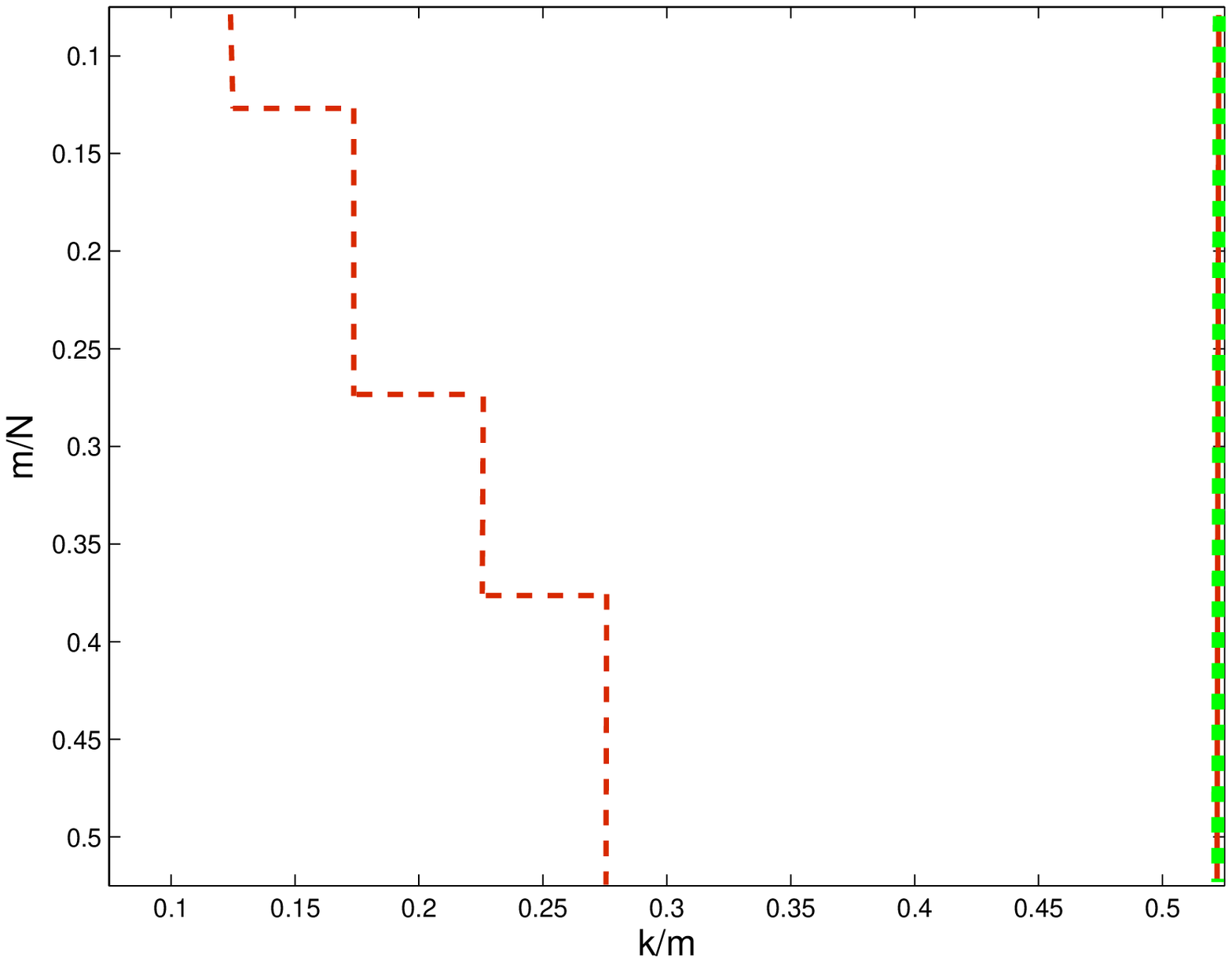}}
	}
	\caption{Phase transition diagrams showing exact recovery rates using weighted $\ell_1$ minimization with weights applied to support estimate sets $\widetilde{T}$ with varying accuracies $\alpha \in \{0.1, 0.3, 0.7, 1\}$. Here, $\w = 1 - \alpha$. The dashed red line corresponds to the empirical 0.85 rate threshold for standard $\ell_1$ minimization. The solid red lines are the 0.85 rate thresholds for weighted $\ell_1$ minimization. For comparison, the dashed green lines correspond to  $m = k + s\log(N/s)$, where $s = (1 + \rho - 2\alpha\rho)k$.}
	\label{fig:PT_wL1}
\end{figure}

We illustrate the phase transitions of weighted $\ell_1$ minimization in Figs.~\ref{fig:PT_wL1} (a)-(d), corresponding to support estimate accuracies of $\alpha \in \{0.1, 0.3, 0.7, 1\}$, respectively. The solid red lines indicate the 0.85 empirical recovery rate thresholds for each of the weighted $\ell_1$ problems. Notice that the recovery thresholds of the weighted $\ell_1$ problems are consistently to the right of the standard $\ell_1$ recovery threshold (dashed red line) for values of $\alpha$ considered. For illustration purposes, we also include the threshold given by $m = k + s\log(N/s)$ (dashed green line) in the figure. This threshold is essentially the best that one can hope to achieve to guarantee exact recovery using weighted $\ell_1$ minimization, as discussed in Section \ref{sec:GaussMat}.

\bibliographystyle{IEEEtran}
\bibliography{wL1}

\end{document}